\documentclass[table,onecolumn]{IEEEtran}

\usepackage{setspace}
\usepackage{authblk}

\usepackage{tikz,amsmath}
\usepackage{amssymb}
\usepackage{amsthm}
\usepackage{hyperref}
\usepackage{graphicx}
\usepackage{xcolor}
\usepackage{balance}
\usepackage{mathtools}
\DeclarePairedDelimiter\ceil{\lceil}{\rceil}

\newtheorem{theorem}{Theorem}[section]

\newtheorem{lemma}[theorem]{Lemma}
\newtheorem{conjecture}[theorem]{Conjecture}

\newcommand{\E}{\mathbb{E}}
\newcommand{\PP}{\mathbb{P}}

\newcommand\hlight[1]{\tikz[overlay, remember picture,baseline=-\the\dimexpr\fontdimen22\textfont2\relax]\node[rectangle,fill=gray!50,rounded corners,fill opacity = 0.2,draw,thick,text opacity =1] {$#1$};} 

\usepackage{mdframed}

\title{A Proof of Entropy Minimization for Outputs in Deletion Channels via Hidden Word Statistics}

\author[1]{Arash Atashpendar \thanks{arash.atashpendar@uni.lu}}
\author[2]{David Mestel \thanks{david.mestel@cs.ox.ac.uk}}
\author[2]{A. W. Roscoe \thanks{bill.roscoe@cs.ox.ac.uk}}
\author[1]{Peter Y. A. Ryan \thanks{peter.ryan@uni.lu}}
\affil[1]{SnT, University of Luxembourg, Luxembourg}
\affil[2]{Department of Computer Science, University of Oxford, Oxford, UK}

\begin{document}

\maketitle

\begin{abstract}
From the output produced by a memoryless deletion channel from a uniformly random input of known length $n$, one obtains a posterior distribution on the channel input.  The difference between the Shannon entropy of this distribution and that of the uniform prior measures the amount of information about the channel input which is conveyed by the output of length $m$, and it is natural to ask for which outputs this is extremized. This question was posed in a previous work, where it was conjectured on the basis of experimental data that the entropy of the posterior is minimized and maximized by the constant strings $\texttt{000\ldots}$ and $\texttt{111\ldots}$ and the alternating strings $\texttt{0101\ldots}$ and $\texttt{1010\ldots}$ respectively.
In the present work we confirm the minimization conjecture in the asymptotic limit using results from \textit{hidden word statistics}.  We show how the analytic-combinatorial methods of Flajolet, Szpankowski and Vall\'ee for dealing with the \textit{hidden pattern matching} problem can be applied to resolve the case of fixed output length and $n\rightarrow\infty$, by obtaining estimates for the entropy in terms of the moments of the posterior distribution and establishing its minimization via a measure of autocorrelation.
\end{abstract}


\section{Introduction}

This work was originally motivated by an analysis of prepare-and-measure based quantum key distribution (QKD) protocols \cite{ryan2013enhancements}, which suggested some simple changes aimed at reducing leakage of key material and improving the final key rate. These changes also included a modification of the quantum bit error rate (QBER) estimation that gave rise to an independent information theory problem that was studied in \cite{atashpendar2015information, atashpendar2018clustering}. We will abstract away from the details of the original context and simply state the problem as an analysis of entropy extremizing outputs in deletion channels.

More formally, the problem can be described as follows. A random bit string $y$ of length $n$ emitted from a memoryless source is transmitted via an i.i.d. deletion channel such that a shorter bit string $x$ of length $m$ ($m \le n$) is received as a subsequence of $y$, after having been subject to $n-m$ deletions. Consequently, the order in which the remaining bits are revealed is preserved, but the exact positions of the bits are not known. Given a subsequence $x$, the question is to find out how much information about $y$ is revealed. More specifically, the quantity that we are interested in is the conditional entropy \cite{cover2012elements} computed over the set of candidate supersequences upon observing $x$, i.e., $H(Y|X=x)$ where $Y$ is restricted to the set of compatible supersequences as explained below.

This information leakage is quantified as the drop in entropy \cite{shannon2001mathematical} for a fixed $x$ according to a weighted set of its compatible supersequences, referred to as the \textit{uncertainty set}. The uncertainty set, denoted by $\Upsilon_{n,x}$, contains all the supersequences that could have given rise to $x$ upon $n-m$ deletions. The weight distribution used in the computation of entropy is given by the number of occurrences or embeddings of a fixed subsequence in its compatible supersequences, i.e., the number of distinct ways $x$ can be extracted from $y$ upon a fixed number of deletions, denoted by $\omega_x(y)$.

The entropy extremization question was first investigated in \cite{atashpendar2015information} where it was conjectured on the basis of experimental data that the entropy conditioned on the observation of a fixed output is minimized and maximized by the uniform (\texttt{111...1}) and the alternating (\texttt{1010...}) bit strings, respectively. In a follow-up work \cite{atashpendar2018clustering}, in addition to studying a series of related combinatorial problems, the authors also provided an analysis of the same information theory problem proving the entropy minimization conjecture for the special cases of single and double deletions, i.e., $m=n-1$ and $m=n-2$. While the methodology used in \cite{atashpendar2018clustering} depended on showing that any bit string can be transformed into the uniform bit string by successively applying an operation that strictly decreases the entropy, here we adopt an entirely different approach based on some key theorems proven in the works of Flajolet, Szpankowski and Vall\'ee \cite{flajolet2006hidden} on \textit{hidden word statistics}. More precisely, we rely on the fact that the distribution of subsequence embeddings asymptotically tends to a Gaussian to obtain estimates for the entropy based on the moments of the posterior distribution. A crucial quantity for establishing the limiting case of entropy minimization is a measure of autocorrelation that is used in estimating the variance. The entropy minimization result ultimately follows from a maximization of this autocorrelation coefficient by the uniform string. The number of runs and their respective lengths in $x$ strings play a central role in the distribution of subsequence embeddings, and in turn, in the corresponding entropy. While this property was already hinted at in \cite{atashpendar2015information}, and directly used in the entropy minimization proof for single and double deletions in \cite{atashpendar2018clustering}, our numerical results in this work indicate that the autocorrelation coefficient captures this run-dependent entropy ordering perfectly.

Although this problem was first encountered while investigating some of the classical sub-protocols in quantum key exchange, the underlying combinatorial puzzle is closely related to several well-known challenging problems in formal languages, DNA sequencing and coding theory. The common thread shared between the present work and the previous papers in this series \cite{atashpendar2015information,atashpendar2018clustering} can be described as a characterization of the limiting entropic cases of the distribution of subsequence embeddings over candidate bit strings transmitted via a deletion channel. This problem is directly linked to that of enumerating the occurrences of a fixed pattern as a subsequence in a random text, also known as the \textit{hidden pattern matching} problem \cite{flajolet2006hidden}. Moreover, the distribution of the number of times a string $x$ appears as a subsequence of $y$, lies at the center of the long-standing problem of determining the capacity of deletion channels: knowing this distribution would give us a maximum likelihood decoding algorithm for the deletion channel \cite{mitzenmacher2009survey}. In effect, upon receiving $x$, every set of $n-m$ symbols is equally likely to have been deleted. Thus, for a received sequence, the probability that it arose from a given codeword is proportional to the number of times it is contained as a subsequence in the originally transmitted codeword. More specifically, we have $p(y|x) = p(x|y)\frac{p(y)}{p(x)} = \omega_x(y)d^{n-m}(1-d)^m \frac{p(y)}{p(x)}$, with $d$ denoting the deletion probability. Thus, as inputs are assumed to be a priori equally likely to be sent, we restrict our analysis to $\omega_x(y)$ for simplicity.

In the present work, we confirm the entropy minimization conjecture in the asymptotic limit using results from hidden word statistics. To do so, we relate our study to the hidden pattern matching problem investigated in the works of Flajolet et al. \cite{flajolet2006hidden}. We show how their analytic-combinatorial methods can be applied to resolve the case of fixed output length and $n\rightarrow\infty$, by obtaining estimates for the entropy in terms of the moments of the posterior distribution.

\subsection{Results}

We consider the random variable $\Omega_n$, the number of ways of embedding a given output string into a uniformly random input string. Results from hidden word statistics derived by Flajolet et al. \cite{flajolet2006hidden} establish a Gaussian limit law for $\Omega_n$ by showing that the moments of $\Omega_n$ converge to the appropriate moments of the standard normal distribution and determine the mean and variance of the number of embedding occurrences. We use these results to establish the limiting case of the random variable $\Omega_n$ in terms of its variance via an approach that depends intricately on the form of $x$ by incorporating a measure of autocorrelation of $x$. We then relate these results to the original entropy problem to prove the case of maximal information leakage for large $n$.

\subsection{Structure}

We provide an overview of related work in Section \ref{sec:related-work}. In Section \ref{sec:framework}, we introduce some notation and describe the main definitions, models, and building blocks used in our study. We then relate our work to the hidden pattern matching problem in Section \ref{sec:hws-entropy-estimation} and use results from hidden word statistics to prove the entropy minimization conjecture. Finally, we conclude by presenting some open problems in Section \ref{sec:conclusions}.

\section{Related Work}\label{sec:related-work}

Combinatorial problems related to subsequences and supersequences crop up pervasively in a wide variety of contexts such as formal languages, coding theory, computer intrusion detection and DNA sequencing to name a few. Despite their prevalence in a wide range of disciplines, they still represent a rich area of research offering a variety of open questions. For example, in the realm of stringology and formal languages, the problem of determining the number of distinct subsequences obtainable from a fixed number of deletions, along with closely related problems, have been studied extensively in  \cite{chase1976subsequence,flaxman2004strings,hirschberg1999bounds,hirschberg2000tight}. It is worth pointing out that the same entropy extremizing strings conjectured in \cite{atashpendar2015information, atashpendar2018clustering} and characterized in the present work, have been shown to lead to the minimum and maximum number of distinct subsequences, respectively. The problems of finding shortest common supersequences (SCS) and longest common subsequences (LCS) represent two other well-known NP-hard problems \cite{jiang1995approximation,middendorf1995finding,middendorf2004combined} that involve subproblems similar to our work. Finally, devising efficient algorithms based on dynamic programming for counting the number of occurrences of a subsequence in DNA sequencing is yet another important and closely related line of research \cite{rahmann2006subsequence,elzinga2008algorithms}.

In coding theory, similar long-standing problems have been studied for several decades, and yet many problems still remain elusive in the context of insertion and deletions channels. This includes designing optimal coding schemes and determining the capacity of deletion channels, both of which incorporate the same underlying combinatorial problem addressed in the present work. Considering a finite number of insertions and deletions for designing correcting codes for synchronization errors \cite{ullman1967capabilities,swart2003note,kanoria2013optimal} and reconstructing the original string from a fixed subsequence \cite{graham2015binary} represent two specific and related research areas. More recent works on the characterization of the number of subsequences obtained via the deletion channel \cite{sala2013counting,sala2015three,liron2015characterization}, e.g., in terms of the number of runs in a string, show great overlap with the present work and the clustering techniques developed in the finite-length analysis of the same problem in \cite{atashpendar2018clustering}. This also includes a graph-theoretic approach for deletion correcting codes \cite{cullina2012coloring}, which is also closely related to the finite-length analysis in \cite{atashpendar2018clustering}. Another important body of research in this area is dedicated to developing bounding techniques \cite{ordentlich2014bounding} and deriving tight bounds on the capacity of deletion channels \cite{diggavi2007capacity,kalai2010tight,rahmati2013bounds,cullina2014improvement}.

Despite being of interest to various disciplines, the problem of determining the number of occurrences or embeddings of a fixed subsequence in random sequences had not been comprehensively studied until Flajolet, Szpankowski and Vall\'ee gave a complete characterization of the statistics of this problem in the asymptotic limit \cite{flajolet2006hidden}. However, the state-of-the-art in the finite-length domain remains rather limited in scope. More precisely, the distribution of subsequence embeddings constitutes a central problem in coding theory, with a maximum likelihood decoding argument, which represents the holy grail in the study of deletion channels. A comprehensive survey, which among other things, outlines the significance of figuring out this particular distribution, was given by Mitzenmacher in \cite{mitzenmacher2009survey}.

Another highly relevant area of research worth mentioning corresponds to the work of Gentleman and Mullin \cite{gentleman1989distribution} in the context of DNA sequencing, which seems to have gone largely unnoticed by the other communities. Their analysis revolves around the characterization of the distribution of the frequency of occurrence of nucleotide subsequences based on their overlap capabilities. The overlap capability of a subsequence is central to their approach for deriving the expectation and the variance of the distribution. This is very much in line with the notion of autocorrelation used by Flajolet et al. in \cite{flajolet2006hidden} almost fifteen years later. A similar study based on \cite{gentleman1989distribution}, also related to nucleotide subsequences, is available at \cite{wu2005distributions}.

Although the finite-length domain still remains quite elusive, here we make use of an asymptotic description of the statistics of hidden patterns given by Flajolet et al. in \cite{flajolet2006hidden} to establish the minimal entropy conjecture. To the best of our knowledge, an analysis focusing on a characterization of the mutual information for the deletion channel \cite{drmota2012mutual} is the only study that directly applies results from hidden word statistics to an information-theoretic analysis.

\section{Framework}\label{sec:framework}

In this section we first describe the notation and terminology used in our work and then introduce the main concepts and definitions that we will need throughout. We will also review some of the building blocks used in hidden word statistics that will be required for obtaining our results.

\subsection{Subsequence Embeddings and Entropy}

\paragraph{Notation} We use the notation $[n] = \{1, 2, \dotsc, n\}$ and $[n_1,n_2]$ to denote the set of integers between $n_1$ and $n_2$; individual bits from a string are indicated by a subscript denoting their position, starting at $1$, i.e., $y = (y_{i})_{i \in [n]} = (y_1, \dotsc, y_n) $. We denote by $|S|$ the size of a set $S$ and the length of a binary string. We also introduce the following notation: when dealing with binary strings, $[a]^k$ means $k$ consecutive repetitions of $a \in \{0, 1\}$. Throughout, we use $h(s)$ to denote the Hamming weight of the binary string $s$.

\paragraph{Probabilistic Model and Alphabet} We consider a memoryless i.i.d. source that emits symbols of the input string (supersequence), drawn independently from the binary alphabet $\Sigma = \{0, 1\}$. Let $\Sigma^n$ denote the set of all $\Sigma$-strings of length $n$ and $p_\alpha$ the probability of the symbol $\alpha \in \Sigma$ being emitted. For a given input length $n$, a random text is drawn from the binary alphabet according to the product probability on $\Sigma^n$: $p(y) \equiv p(y_1 \ldots y_ n) = \prod_{i=1}^n p_{y_i} = 2^{-n}$. The probability of a subsequence of length $m$ is defined in a similar manner.

\paragraph{Subsequences and Supersequences} Given $x \in \Sigma^m$ and $y \in \Sigma^n$, let $x = x_1 x_2 \cdots x_m$ denote a subsequence obtained from a supersequence $y = y_1 y_2 \cdots y_n$ with a set of indexes $1 \le i_1 < i_2 < \cdots < i_m \le n$ such that $y_{i_1} = x_1, y_{i_2} = x_2, \dotsc, y_{i_m} = x_m$. Subsequences are obtained by deleting characters from the original string and thus adjacent characters in a given subsequence are not necessarily adjacent in the original string.

\paragraph{Projection Masks} We define $y_\pi = (y_i)_{i \in \pi} = x$ to mean that the string $y$ filtered by the mask $\pi$ gives the string $x$. Let $\pi$ denote a set of indexes $\{j_1, \dotsc, j_m\}$ of increasing order that when applied to $y$, yields $x$, i.e., $x = y_{j_1} y_{j_2} \cdots y_{j_m}$ and $1 \le j_1 < j_2 \cdots j_m \le n$.

\paragraph{Compatible Supersequences} We define the \emph{uncertainty set}, $\Upsilon_{n,x}$, as follows. Given $x$ and $n$, this is the set of $y$ strings that could project to $x$ for some projection mask $\pi$.
\begin{align*}
\Upsilon_{n,x}:=\{y \in \{0,1\}^n: (\exists \pi) [y_\pi=x] \}
\end{align*}

\paragraph{Number of Masks or Subsequence Embeddings} Let $\omega_x(y)$ denote the number of distinct ways that $y$ can project to $x$:
\[
\omega_x(y) := | \{ \pi \in \mathcal{P}([n]): y_\pi = x \} |
\]
we refer to the number of masks associated with a pair $(y, x)$ as the weight of $y$, i.e., the number of times $x$ can be embedded in $y$ as a subsequence. Moreover, we use $\Omega_n(x)$ to denote the number of occurrences of a given subsequence $x$ in a random text of length $n$ generated by a memoryless source.

\paragraph{Entropy} For a fixed subsequence $x$ of length $m$, the underlying weight distribution used in the computation of the entropy is defined as follows. Upon receiving a subsequence $x$, we consider the set of compatible supersequences $y$ of length $n$ (denoted by $\Upsilon_{n,x}$) that can project to $x$ upon $n-m$ deletions. Every $y \in \Upsilon_{n,x}$ is assigned a weight given by its number of masks $\omega_x(y)$, i.e., the number of times $x$ can be embedded in $y$ as a subsequence. We consider the conditional Shannon entropy $H(Y|X=x)$ where $Y$ is confined to the space of compatible supersequences $\Upsilon_{n,x}$. The total number of masks in $\Upsilon_{n,x}$ is given by
\begin{equation}
\mu_{n,m}=\binom{n}{m} \cdot 2^{n-m}
\end{equation}
Thus, forming the normalized weight distribution
\begin{equation}\label{eq:subseq-prob-distribution}
P_x=\Bigg\{\frac{\omega_x(y_1)}{\mu_{n,m}}, \ldots, \frac{\omega_x(y_n)}{\mu_{n,m}} \Bigg\}.
\end{equation}
where $P(Y = y | X = x)$ is given by
\begin{align*}
P(Y = y | X = x) &= \frac{P(Y = y \wedge X = x)}{P(X = x)} = \frac{P(X = x| Y = y) \cdot P(Y = y)}{P(X = x)}
= \frac{\frac{|\{ \pi: \pi(y)=x \}|}{\binom{n}{m}}2^{-n}}{P(X =x)} \\
& = \frac{\omega_x(y)2^{-n}}{\binom{n}{m}P(X =x)} = \frac{\omega_x(y)2^{-n}}{\binom{n}{m}\sum_{y'}P(Y = y')P(X = x | Y = y')} \\
& = \frac{\omega_x(y)2^{-n}}{\binom{n}{m}\sum_{y'}\frac{\omega_x(y')}{\binom{n}{m}}2^{-n}} = \frac{\omega_x(y)}{\sum_{y'}\omega_x(y')} = \frac{\omega_x(y)}{\mu_{n,m}}
\end{align*}
Finally, for simplicity we use $H_n(x)$ throughout this work to refer to the entropy of a distribution $P$ corresponding to a subsequence $x$ as defined below
\begin{equation}\label{eq:shannon-entropy}
H_n(x) = -\sum_i p_i \cdot \log_2(p_i)
\end{equation}
where $p_i$ is given by
\begin{equation*}
p_i = \frac{\omega_x(y_i)}{\mu_{n,m}}.
\end{equation*}

\subsection{Building Blocks from Hidden Word Statistics}

In the terminology of hidden word/pattern statistics, the same problem of determining the number of distinct embeddings of a subsequence in a supersequence is referred to as the ``hidden pattern matching'' problem. Here we review the most relevant concepts introduced in the work of Flajolet, Szpankowski and Vall\'ee \cite{flajolet2006hidden}.

\paragraph{Hidden Patterns and Constraints} Let $\mathcal{W} = w_1, w_2, \ldots, w_m$ denote the pattern or subsequence obtained from the text $T_n = t_1, t_2, \ldots, t_n$, and let $\mathcal{D} = (d_1, \ldots , d_{m-1})$ be an element of $(\mathbb{N}^+ \cup \{ \infty \})^{m-1}$. The pattern matching problem is determined by a pair $(\mathcal{W}, \mathcal{D})$, called a ``hidden pattern'' specification, i.e., a subsequence pattern $\mathcal{W}$ along with an additional set of constraints $\mathcal{D}$ on the indices $i_1, i_2, \ldots, i_m$.  If an occurrence in the form of an $m$-tuple $S=(i_1, i_2, \ldots, i_m)$ with $(1 \le i_1 < i_2 < \ldots < i_m)$ satisfies the constraint $\mathcal{D}$, i.e., $i_{j+1} - i_j \le d_j$, it is then considered to be a valid mask or a position. In essence, the notion of constraints models the existence of gaps between the embeddings of the symbols of a subsequence in a random text. In other words, the analysis considers the number of occurrences of a subsequence as embeddings that satisfy a specific set of distance constraints.

Moreover, let $\mathcal{P}_n(D)$ be the set of all positions subject to the separation constraint $\mathcal{D}$, satisfying $i_m \le n$. Let also $\mathcal{P}(D) = \bigcup_n \mathcal{P}_n(D)$. This allows us to view the number of occurrences $\Omega$ of a subsequence $w$ in text $T$ subject to the constraint $\mathcal{D}$ as a sum of characteristic variables
\begin{equation}
\Omega(T) = \sum_{I \in \mathcal{P}_{|T|}(\mathcal{D})} X_I(T), \quad \text{with} \quad X_I(T) := [[ \mathrm{w}\, \text{occurs at position } I \text{ in } T ]].
\end{equation}
with $[[B]]$ being 1 if the property $B$ holds and 0, otherwise.

The two ends of the spectrum in this model are given by the following. The \emph{fully unconstrained case} is modelled by $\mathcal{D} = (\infty, \ldots, \infty)$; whereas the \emph{constrained problem} is modelled by the case where all $d_j$ are finite. Our study is only concerned with the former, namely the fully unconstrained problem, as we allow an arbitrary number of symbols in between the gaps.

\paragraph{Blocks} A given pattern $x$ is broken down into $b$ independent subpatterns that are called blocks, $x_1, x_2, \ldots, x_b$. The quantity denoted by $b$ is defined as the number of unbounded gaps (the number of indices $j$ for which $d_j = \infty$) plus 1, which is also referred to as the number of blocks. The two extreme cases, namely the fully unconstrained and the fully constrained problem, are thus described by $b=m$ and $b=1$, respectively. For the purpose of our study, we always assume $b=m$. Collections of blocks are then used to form an \emph{aggregate}, which describes the interval of indices that marks a block, the first and last index in an interval. One of the main uses of blocks and aggregates is to model the fact that masks and occurrences of a subsequence can overlap with each other by quantifying the extent to which such overlaps can occur. However, as we are only interested in the fully unconstrained case, covering the notion of aggregates goes beyond the scope of our work. The reader is encouraged to refer to \cite{flajolet2006hidden} for a more complete and detailed presentation of these concepts.

\section{Estimating Entropy using Hidden Word Statistics}\label{sec:hws-entropy-estimation}

We now revisit the original entropy problem and provide an analysis in the asymptotic limit by considering the case of fixed output length $m$ and $n \rightarrow \infty$. This allows us to apply results from hidden pattern statistics to establish the limiting case of minimal entropy. The probabilistic aspects of the statistics of hidden patterns were quantified by Flajolet et al. in an extensive study \cite{flajolet2006hidden}, which was originally motivated by intrusion detection in computer security. Among other things, they showed that the random variable $\Omega_n$ asymptotically tends to a Gaussian. We relate our work to their study and incorporate two key theorems related to hidden patterns to establish the limiting case of minimal entropy via a notion of autocorrelation associated with subsequences.

\subsection{Hidden Word Statistics}

In \cite{flajolet2006hidden}, it is shown that for fixed short strings of length $m$ as $n \to \infty$, the dominant contribution to the moments comes from configurations where the positions of the short strings are minimally intersecting. We will briefly describe the approach used in \cite{flajolet2006hidden}.

For a position $I$ (that is, a subset of $[n]$ of size $m$), let $X_I$ denote the indicator of the event that the long string restricted to $I$ matches the short string. Let $Y_I = X_I - \mathbb{E}(X_I) = X_I - 2^{-m}$. Then $X = \Omega - E = \sum_I Y_I$, and so

\begin{equation}
\mathbb{E}(X^r) = \sum_{I_1, \ldots, I_r} \mathbb{E}(Y_{I_1} \ldots Y_{I_r}).
\end{equation}

Now let $\mathcal{O}_r$ be the combinatorial class consisting of pairs $((I_1, \ldots, I_r), T)$, where the $I_j$ are positions and $T$ is a ``text'' (i.e. a string of length $> m$), taken with weight $Y_{I_1}(T) \ldots Y_{I_r}(T)2^{-|T|}$. Now if $O_r (z)$ is the generating function of $O_r$ with this weighting, we have

\begin{equation}
[z^n]O_r(z) = \sum_{|T|=n}\sum_{I_1,\ldots,I_r} Y_{I_1}(T) \ldots Y_{I_r}(T) 2^{-n} = \mathbb{E}(X^r).
\end{equation}
Note that $[z^n]O_r(z)$ means the coefficient of $z^n$ in $O_r(z)$.

We can partition $\mathcal{O}_r$ according to the number of points covered by some $I_j$. Let $\mathcal{O}_r^{[p]}$ denote the class of elements in which the number of points covered is $rm-p$. Note that if $I_1$ does not intersect with any other $I_j$, then $Y_1$ is independent of $Y_2, \ldots, Y_r$ and so $\mathbb{E}(Y_1 \ldots Y_r) = 0$, so contributions only come from families where each position intersects some other position. Such families are called ``friendly'' and require in particular $p \ge \ceil*{r/2}$.

To obtain the generating function for $\mathcal{O}_r^{[p]}$, we apply a combinatorial isomorphism to group together all the covered points of intersection, so that we have
\begin{equation}
\mathcal{O}_r^{[p]} \cong (\{ 0,1 \}^*)^{rm-p+1} \times \mathcal{B}_r^{[p]},
\end{equation}
where $\mathcal{B}_r^{[p]}$ is the subset of $\mathcal{O}_r^{[p]}$ which is \emph{full}, that is, for which the set of covered points is contiguous. We thus have
\begin{equation}
O_r^{[p]}(z) = \Big( \frac{1}{1-z} \Big)^{rm-p+1} \times B_r^{[p]}(z).
\end{equation}
Since the analysis in \cite{flajolet2006hidden} considers a fixed short string of length $m$ as $n \to \infty$, it is enough to observe that $B_r^{[p]}(z)$ is some \emph{fixed} polynomial, because one can then easily show that the coefficient $[z^n]O_r^{[p]} = O(n^{rm-p})$. This means that however fast the coefficients of $B_r^{[p]}(z)$ grow as $p$ grows, for large enough $n$, the minimal-$p$ term will dominate.

\subsection{Establishing Entropy Minimization  via Hidden Word Statistics}

We will rely on the fact that the distribution of $\Omega_n$ asymptotically tends to a Gaussian and use a measure of autocorrelation defined for subsequences to obtain estimates for the entropy in terms of the moments of the posterior distribution. Indeed, the underlying probability distribution in our original entropy analysis coincides with that of the so-called \emph{hidden pattern matching} problem in which one searches for the number of occurrences of a given \emph{pattern}\footnote{The words ``subsequence'' and ``pattern'' are used interchangeably.} $\mathcal{W}$, as a subsequence in a random text $T$ of length $n$ generated by a memoryless source. More precisely, given that $\Omega_n \sim \mathcal{N}(\mu, \sigma^2)$, we will analyze how the mean and the variance of the distribution change for different $x$ strings in order to resolve the limiting case of minimal entropy exhibited by the uniform string $[0]^m$.

The probabilistic analysis done in \cite{flajolet2006hidden} relies on a description of the structures of interest in formal languages, involving a joint use of combinatorial-enumerative techniques and analytic-probabilistic methods. This approach enables a systematic translation of the combinatorial problem into generating functions. The essential combinatorial-probabilistic features of the problem, such as variance coefficients and a notion of autocorrelation, are derived by using an asymptotic simplification made possible by the use of the singular forms of generating functions. For an extensive and complete coverage of these techniques, we refer the reader to \cite{flajolet2009analytic,sedgewick2013introduction}.

In our work, we will mainly make use of two fundamental theorems presented in \cite{flajolet2006hidden}. The first theorem states that $\Omega_n$ asymptotically tends to a Gaussian, while the second theorem provides analytic expressions for its moments, i.e., the expectation and the variance of $\Omega_n$. Another equally important result that we will use to distinguish between two different subsequences of length $m$ is a measure of autocorrelation that depends intricately on the exact form of $x$. Given that the mean (Eq. \ref{eq:omega-mean}) is constant for all $x$ strings of equal length, the autocorrelation factor, incorporated in the variance coefficient, allows us to differentiate between two subsequences in that it is the only term  that depends on the form of $x$, with all other terms in Eq. \ref{eq:omega-variance} being only a function of $n$ and $m$.

\subsection{Distribution of Subsequence Embeddings in the Asymptotic Limit}

The plots given in Fig.\ref{figure:gaussian-omega} illustrate the convergence of the distribution of $\Omega_n$ to a Gaussian for the subsequence $x=\texttt{01}$ and increasing values of $n$. As already mentioned, the distribution of subsequence embeddings tending to a Gaussian in the asymptotic limit is of particular significance for our work given that $\Omega_n$ is precisely the random variable associated with the weights of the supersequences in $\Upsilon_{n,x}$ for the computation of entropy.

\begin{figure}[ht!]
	\centering
	\includegraphics[scale=0.7]{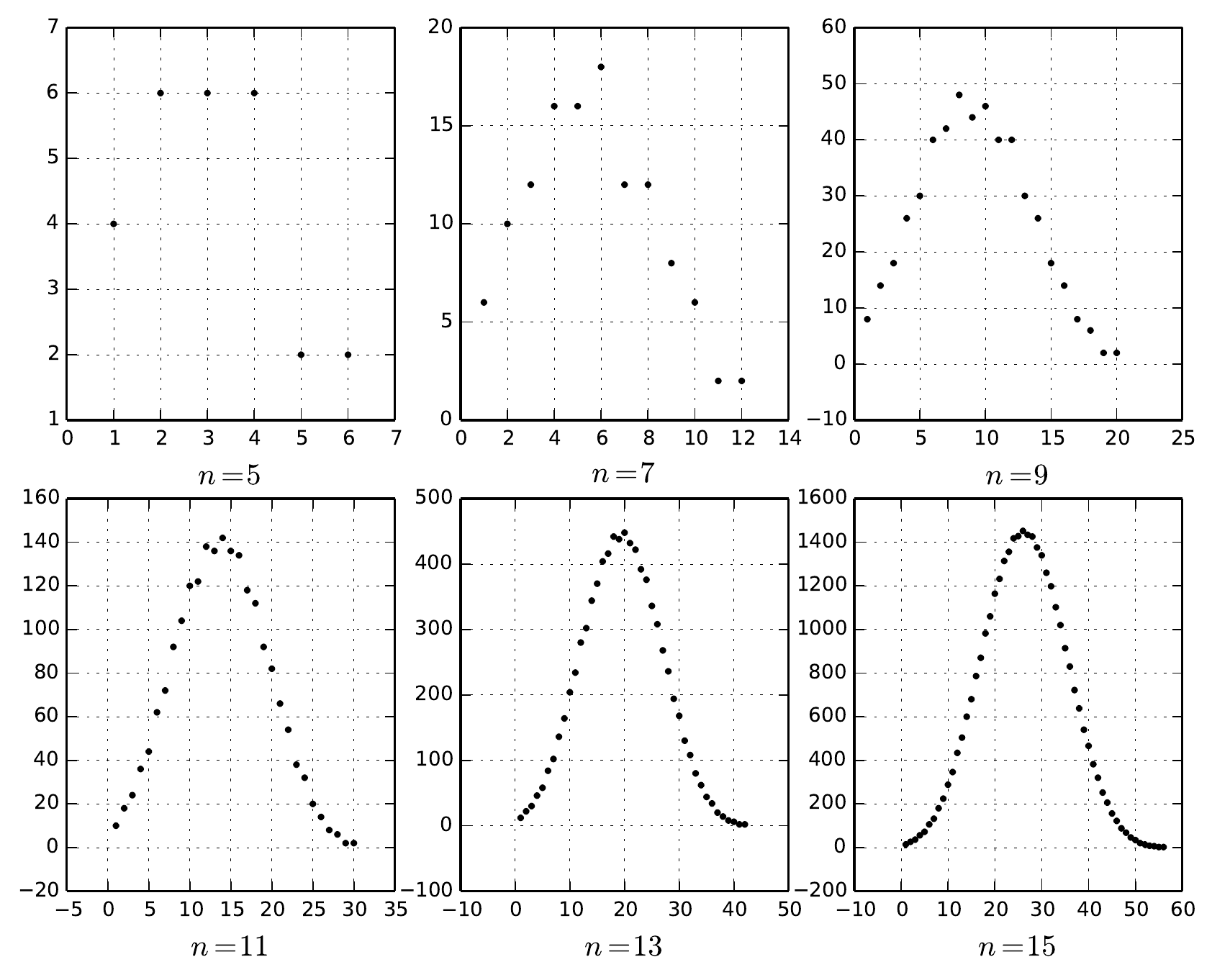}
	\caption{Frequency distribution of $\Omega_n$ converging to a Gaussian for $x=\texttt{01}$ and $n= 5\ldots15$}
	\label{figure:gaussian-omega}
\end{figure}

In the following, we first present the analytic expressions satisfying the mean and the variance of the number of occurrences $\Omega_n$ and adapt them to the parameters of our problem. We then characterize the limiting case of minimal entropy exhibited by the uniform string, i.e. $x = [0/1]^m$, via a notion of autocorrelation coefficient incorporated in the variance.

\subsubsection{Moments and Convergence}

The results provided here have been sourced from \cite{flajolet2006hidden} and adapted to the specific parameters of our problem, i.e., we consider the fully unconstrained setting, restricted to the binary alphabet. For all $x$ strings of length $m$, the mean is constant and therefore, we mainly focus on the variance.
\begin{theorem}\label{thm:hws_moments}\cite{flajolet2006hidden}
The mean and the variance of the number of occurrences $\Omega_n$ of a subsequence $x$ for $p_\alpha = 0.5$, subject to constraint $\mathcal{D}=(\infty, \ldots, \infty)$, and thus $b=m$, are given by
\begin{equation}\label{eq:omega-mean}
	\mathbb{E}[\Omega_n] = \frac{2^{-m}}{m!} n^m \left(1 + O\left(\frac{1}{n}\right)\right)
\end{equation}
\begin{equation}\label{eq:omega-variance}
	\mathbb{V}[\Omega_n] = \frac{2^{-2m}}{(2m-1)!} \kappa^2(x) n^{2m-1} \Bigg(1 + O\Bigg(\frac{1}{n}\Bigg)\Bigg),
\end{equation}
where the autocorrelation $\kappa^2(x)$ is defined by
\begin{equation}\label{eq:kappa}
	\kappa^2(x) := \sum_{1 \le r,s \le m} \binom{r+s-2}{r-1} \binom{2m-r-s}{m-r} [[x_r = x_s]].
\end{equation}
Note that $[[P]]$ denotes the indicator function of the property $P$ (so $[[P]] = 1$ if $P$ holds and $0$ otherwise).
\end{theorem}
\begin{theorem}\label{thm:hws_gaussian}\cite{flajolet2006hidden}
\begin{equation}
X_n:=\frac{\Omega_n-\mathbb{E}(\Omega_n)}{\sqrt{\mathbb{V}(\Omega_n)}}
\end{equation}
converges in measure to a standard normal distribution.
\end{theorem}

We encapsulate the multiplicands in the definition of $\kappa^2(x)$ into matrices, viewing the indicator function as a mask on the matrix of binomial coefficients.  Let $\mathcal{B}$ be the matrix representing the indicator function $\mathcal{B}_{r,s}:=[[x_r=x_s]]$, and let $\mathcal{M}$ be the matrix of binomial coefficients
\[
\mathcal{M}_{r,s} = \binom{r+s-2}{r-1} \binom{2m-r-s}{m-r}.
\]
Write $\mathcal{R}=\mathcal{B}\circ\mathcal{M}$, the Hadamard or elementwise product of $\mathcal{B}$ and $\mathcal{M}$, for the result of applying the mask $\mathcal{B}$ to the matrix $\mathcal{M}$.  We then have an equivalent formulation of equation \eqref{eq:kappa}, namely
\[
\kappa^2(x) = \sum_{r=1}^m \sum_{s=1}^m \mathcal{R}_{r,s}.
\]

\subsubsection{Autocorrelation}

It is worthwhile to provide some explanation of the combinatorial meaning of the autocorrelation coefficient $\kappa^2$ derived in \cite{flajolet2006hidden}, in view of its significance in the analysis that follows.

The coefficient $\kappa^2$ is related to a generalization of the autocorrelation polynomial originally introduced for classical string matching by Guibas and Odlyzko \cite{guibas1981periods,guibas1981string}.  The variance of $\Omega_n$ is determined by the probability that a random pair of $m$-subsets of a random long string are \emph{both} matches for the short string, and how this compares to the square of the corresponding probability for a single $m$-subset.

Analytic-combinatorial methods show that the dominant contribution for large $n$ comes from pairs which overlap in only a single position, so computing the variance amounts to counting the number of triples consisting of a long string and a pair of $m$-subsets intersecting in precisely one location such that both are matches for the short string.  Grouping the chosen locations together introduces a constant factor of $\binom{n}{2m-1}2^{n-(2m-1)}$, and so it suffices to count the number of ways to interleave two copies of the short string, with a single intersection.  This quantity is the autocorrelation coefficient $\kappa^2(x)$.

Explicitly, $\mathcal{M}_{r,s}$ is the number of combinations with the $r^{th}$ location of the first set meeting the $s^{th}$ location of the second: $\binom{r+s-2}{r-1}$ is the number of interleavings of the $r-1$ and $s-1$ locations before this, and $\binom{2m-r-s}{m-r}$ the number of interleavings of the $m-r$ and $m-s$ locations after.

\subsection{Maximal Autocorrelation}

We now study the extremization of the variance of $\Omega_n$ by analyzing the extreme values of the autocorrelation $\kappa^2$. Here we consider the all-0s and all-1s strings (x=$[0]^m$ and $[1]^m$), for which the autocorrelation matrix contains $m^2$ 1's: $\forall i,j \in \{1 \ldots m\}: x_i = x_j$.
\begin{theorem}\label{theorem:maximal-variance}
Let $x$ be a string of length $m$. Then
\begin{equation}\label{eq:kappa-max}
\kappa^2(x) \leq \kappa^2\left([0]^m \right) 
= \kappa^2\left([1]^m \right) = m\binom{2m-1}{m}.
\end{equation}
\end{theorem}
\begin{proof}
Since $\mathcal{M}$ is independent of the form of $x$, we focus only on the indicator matrix $\mathcal{B}$. It is clear that the constant $x$ strings comprising all 0's and all 1's are the unique strings that result in an all-ones masking matrix $\mathcal{B}$. Consequently,  $\kappa^2([0/1]^m)$ includes all of the $m^2$ terms involved in $\mathcal{M}$ and thus attains its maximal value, i.e., $\kappa^2([0/1]^m) = \sum_{1 \le r,s \le m} \mathcal{R}_{r,s} =  \sum_{1 \le r,s \le m} \mathcal{M}_{r,s}$, hence Eq. \ref{eq:kappa-max}.
\end{proof}

The alternating $x$ string $x=\texttt{1010...}$ appears to lie at the other end of the entropy spectrum. While the proof for the maximization of the autocorrelation coefficient by the all 0's string was rather straightforward, showing its minimization still escapes us. We simply state the minimization as a conjecture.

\begin{conjecture}\label{theorem:minimal-variance}
	The alternating subsequence of length $m$, i.e., $x = \texttt{1010...}$, minimizes the autocorrelation coefficient $\kappa^2$.
\end{conjecture}

\subsection{Entropy}\label{sec:entropy}

We briefly review the results of the entropy analysis in which it is conjectured that the all 0's and the alternating $x$ string, minimize and maximize the entropy, respectively.

The characteristic parameter $\Omega_n$ describes the distribution that underlies the entropy in our problem. The plot given in Fig. \ref{figure:entropy} shows the values of the Min-Entropy ($H_\infty$), the second-order R\'{e}nyi entropy ($R$) and the Shannon entropy ($H$) computed for all $x$ strings of length $m=5$, with $n=8$.
\begin{figure}[ht!]
	\centering
	\includegraphics[scale=0.8]{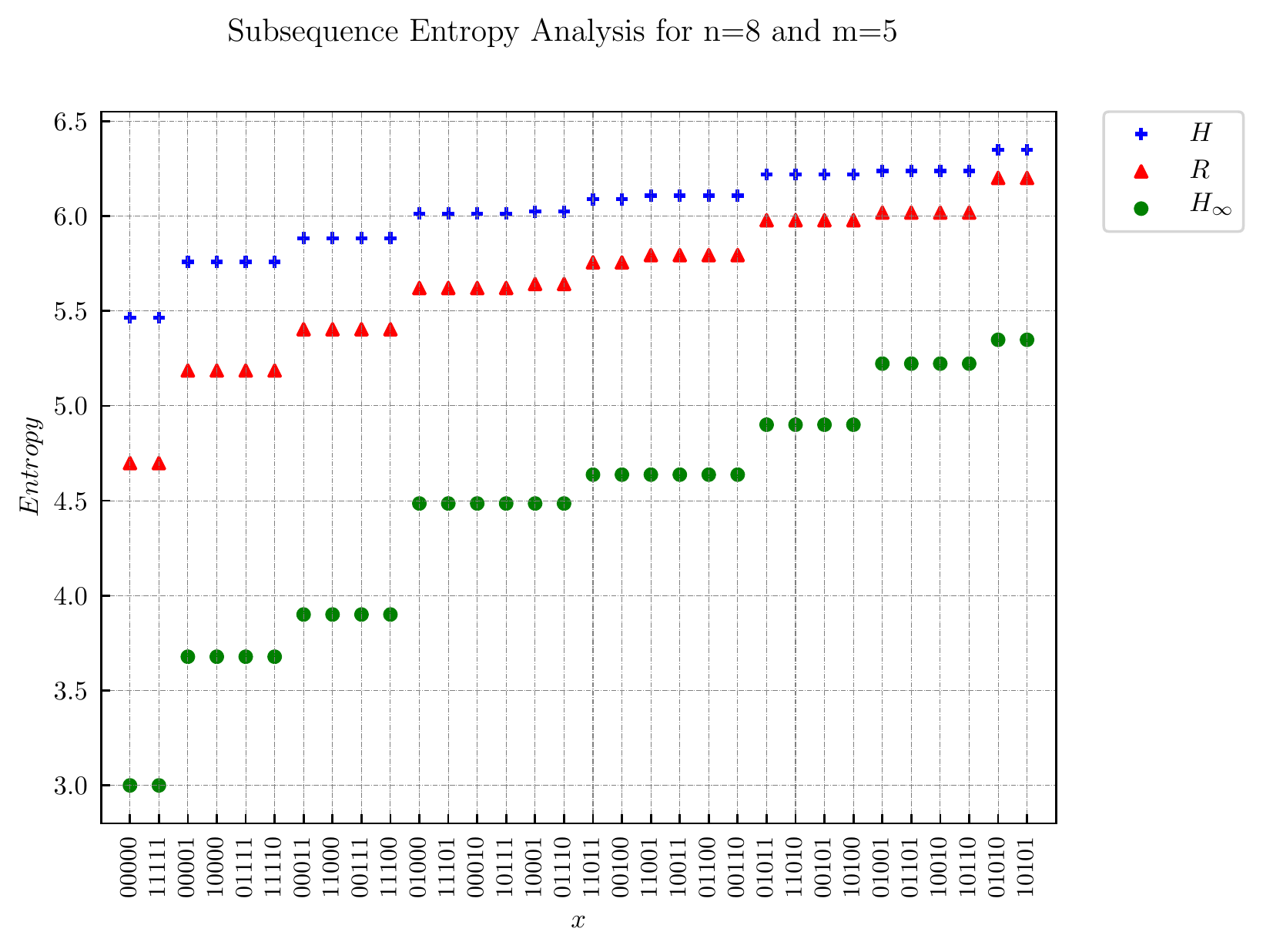}
	\caption{Min-Entropy $H_\infty$, Shannon entropy $H$ and the second-order R\'{e}nyi entropy $R$ vs. $x$}
	\label{figure:entropy}
\end{figure}

\subsubsection{Calculating Entropy From Moments of Distribution}

An equivalent formulation of Theorem \ref{thm:hws_gaussian} (and in fact the form in which it is proved) is that the moments of the (normalized) converge to the corresponding moments of the standard normal distribution:

\begin{lemma}\label{lem:momconv}The moments of the normalized version of $\Omega_n$ converge to the corresponding moments of the standard normal distribution.  That is,
\[\mathbb{E}\left(\left(\frac{\Omega_n - \mathbb{E}(\Omega_n)}{\sqrt{\mathbb{V}(\Omega_n)}} \right)^r \right) \rightarrow \begin{cases}0 &r \mbox{ odd} \\
(r-1)\times (r-3)\times \ldots \times 1 &r \mbox{ even}.\end{cases}\]
\end{lemma}

We have a distribution $\Omega$ and the goal is to estimate $\mathbb{E}(\Omega \mathrm{log} \Omega)$, given the moments of $\Omega$. Let $E = \mathbb{E}(\Omega)$, and let the pdf of $\Omega$ be $f$. By Taylor's theorem, we have

\begin{align}\label{eq:taylor}
    \Omega \log\Omega = E\log E + (\log E + 1)(\Omega - E) + \frac{(\Omega - E)^2}{2E} - \frac{(\Omega-E)^3}{6E^2} +  \mathcal{R}(\Omega),
\end{align}
where $\mathcal{R}(\Omega)$ is the integral form of remainder, i.e.
\[\mathcal{R}(\Omega) = \int_E^\Omega \frac{(\Omega-t)^3}{6t^3} dt.\]

Note that $\mathcal{R}(\Omega)$ is non-negative for all $\Omega$.  Now whenever $\Omega \geq \frac{1}{2}E$, we have
\[\mathcal{R}(\Omega) \leq \int_E^\Omega \frac{8(\Omega-t)^3}{6E^3} dt = \frac{8(\Omega-E)^4}{24E^3}.\]

On the other hand if $\Omega < \frac{1}{2} E$ then 
\begin{align*}\mathcal{R}(\Omega) &= \int^\Omega_{\frac{E}{2}} \frac{(\Omega-t)^3}{6t^3} dt + \int^{\frac{E}{2}}_E \frac{(\Omega-t)^3}{6t^3} dt \\
&\leq \int^\Omega_{\frac{E}{2}} \frac{1}{6} dt + \int_E^{\frac{E}{2}} \frac{8(\Omega-t)^3}{6E^3} dt \\
&\leq \frac{E}{12} + \int_E^\Omega \frac{8(\Omega-t)^3}{6E^3} dt \\
&= \frac{E}{12} + \frac{8(\Omega-E)^4}{24E^3}.
\end{align*}

Hence we have that
\begin{equation}\label{eq:rembound1}
|\E(\mathcal{R})| \leq \frac{1}{12}E\PP\left(\Omega < \frac{1}{2}E\right) + \frac{\E((\Omega-E)^4)}{3E^3}.
\end{equation}

We obtain a Chebychev bound on the first term:
\begin{align*}
    \PP\left(\Omega < \frac{1}{2}E\right) &\leq \PP\left(|\Omega-E| > \frac{1}{2}{E}\right) \\
    &\leq \frac{\E\left((\Omega-E)^4\right)}{\left(\frac{1}{2}E\right)^4} = \frac{16\E\left((\Omega-E)^4\right)}{E^4} 
\end{align*}

Substituting this into (\ref{eq:rembound1}) gives
\begin{equation}\label{eq:rembound}|\E(\mathcal{R})| \leq \frac{5\E((\Omega-E)^4)}{3E^3}.\end{equation}

Hence taking expectations of (\ref{eq:taylor}) gives 
\begin{equation}\label{eq:est}\E(\Omega\log\Omega) = E\log E + \frac{\mathbb{V}(\Omega)}{2E} - \frac{E((\Omega-E)^3)}{6E^2} + \epsilon\left(\frac{5}{3}\frac{\E((\Omega-E)^4)}{E^3}\right),
\end{equation}
where the notation $\epsilon(x)$ means an error term of magnitude at most $x$.

\subsubsection{Minimal Entropy}

We are now in a position to prove the main theorem of this Section, that (for sufficiently large $n$), the entropy is minimized uniquely by the constant strings $[0]^m, [1]^m$.

\begin{theorem}\label{theorem:fixed-k}
For all $m$, there is some $N$ such that for all $n>N$, and any string $x$ of length $m$, we have
\[H_n(x) \geq H_n\left([0]^m\right),\]
with equality only if $x\in \{[0]^m, [1]^m\}$.
\end{theorem}

\begin{proof}
From Eq. \ref{eq:subseq-prob-distribution} and Eq. \ref{eq:shannon-entropy}, we have
\begin{align}\label{eq:omegaest}
\begin{split}
H_n(x) = H(P) &= -\sum_i p_i \log p_i \\
&= -\sum_i \left( \frac{\omega(i)}{\mu} \right) \log \left( \frac{\omega(i)}{\mu} \right) \\
&= \frac{\log\mu}{\mu}\sum_i \omega(i) - \frac{1}{\mu}\sum_i \omega(i)\log\omega(i) \\
&= \log\mu - \E(\Omega\log\Omega)
\end{split}   
\end{align}
Hence it suffices to prove that for sufficiently large n the constant strings maximize $\mathbb{E}\left(\Omega_n \log \Omega_n\right)$.

Note that $E$ depends only on $n$, and not on the form of $x$; by Theorem \ref{thm:hws_moments} we have $E=\Theta(n^m)$.  On the other hand, by the same Theorem we have $\mathbb{V}(\Omega_n) = \frac{2^{-2m}}{(2m-1)!}\kappa^2(x)n^{2m-1}\left(1+O(1/n)\right)$.

Now $\kappa^2$ depends only on the form of $x$ and not on $n$, and by Theorem \ref{theorem:maximal-variance}
it is uniquely maximized by the all-1s/0s strings.  Because $\kappa^2$ is independent of $n$, we therefore also have that the change in $\mathbb{V}(\Omega_n)$ induced by moving away from these strings is $\Theta(n^{2m-1})$, and so it suffices to prove that all of the error terms in \eqref{eq:est} are $o\left(\frac{n^{2m-1}}{E}\right) = o\left(n^{m-1}\right)$.

Now by Lemma \ref{lem:momconv} (combined with the fact that by Theorem \ref{thm:hws_moments} $\mathbb{V}(\Omega_n) = \Theta\left(n^{2m-1}\right)$), we have that $\E((\Omega-E)^3)=o(n^{3m-3/2})$, and $\E((\Omega-E)^4)=O(n^{4m-2})$.  Combining this with the fact that $E=\Theta(n^m)$ yields the required bounds on the errors, and hence the result.
\end{proof}

\subsubsection{Entropy Ordering based on Autocorrelation}\label{sec:numerical-experiments}

Although we have proved the extremal case of minimal entropy in the asymptotic limit for $n \rightarrow \infty$ and fixed output length $m$ via the autocorrelation coefficient $\kappa^2(x)$, it is worth pointing out that our numerical results indicate that $\kappa^2(x)$ predicts the entropy ordering perfectly in the finite length domain as well, i.e., for small and comparable fixed values of $n$ and $m$. An example obtained from empirical data is presented in Table \ref{tab:kappa-entropy} to illustrate the correlation between $H_n(x)$ and $\kappa^2(x)$ for $n=8$ and $m=5$.
\begin{table}[ht!]
	\caption{$x$ strings sorted by $\kappa^2$ in descending order predicts the entropy ordering}
	\label{tab:kappa-entropy}
	\centering
	\begin{tabular}{l*{6}{c}r}
		$x$ & $\kappa^2(x)$ $\downarrow$ & $H(x)$ \\
		\hline
		$11111$ & 630 & 5.4649 \\
		\hline
		$00000$ & 630 & 5.4649 \\
		\hline
		$00001$ & 518 & 5.7581 \\
		$\ldots$ & $\ldots$ & $\ldots$ \\
		\hline
		$11000$ & 486 & 5.8838 \\
		$\ldots$ & $\ldots$ & $\ldots$ \\
		$00010$ & 458 & 6.0132 \\
		$\ldots$ & $\ldots$ & $\ldots$ \\
		$10011$ & 398 & 6.1076 \\
		$\ldots$ & $\ldots$ & $\ldots$ \\
		$01101$ & 366 & 6.2375 \\
		$\ldots$ & $\ldots$ & $\ldots$ \\
		$01010$ & 350 & 6.3498  \\
		\hline
	\end{tabular}
\end{table}

\section{Concluding Remarks}\label{sec:conclusions}

We have provided a proof for the minimization of entropy by the uniform string in the asymptotic limit, i.e., $n\rightarrow\infty$ and fixed output length $m$, using results from hidden word statistics. However, showing the entropy maximization by the alternating string remains an open problem given that a proof establishing the minimization of the autocorrelation coefficient $\kappa^2(\texttt{1010\ldots})$ still escapes us. Beyond establishing this maximization, proving the entropy ordering of $x$ strings determined by $\kappa^2(x)$ for finite $n$ and $m$ represents another open problem.

\bibliographystyle{unsrt}
\bibliography{hws_entropy}

\begin{thebibliography}{10}

\bibitem{ryan2013enhancements}
Peter~YA Ryan and Bruce Christianson.
\newblock Enhancements to prepare-and-measure based qkd protocols.
\newblock In {\em Security Protocols XXI}, pages 123--133. Springer, 2013.

\bibitem{atashpendar2015information}
Arash Atashpendar, AW~Roscoe, and Peter~YA Ryan.
\newblock Information leakage due to revealing randomly selected bits.
\newblock In {\em Security Protocols XXIII}, pages 325--341. Springer, 2015.

\bibitem{atashpendar2018clustering}
Arash Atashpendar, Marc Beunardeau, Aisling Connolly, R{\'e}mi G{\'e}raud,
  David Mestel, AW~Roscoe, and Peter~YA Ryan.
\newblock From clustering supersequences to entropy minimizing subsequences for
  single and double deletions.
\newblock {\em arXiv preprint arXiv:1802.00703}, 2018.

\bibitem{cover2012elements}
Thomas~M. Cover and Joy~A. Thomas.
\newblock {\em Elements of information theory}.
\newblock John Wiley \& Sons, 2012.

\bibitem{shannon2001mathematical}
Claude~E Shannon.
\newblock A mathematical theory of communication.
\newblock {\em ACM SIGMOBILE Mobile Computing and Communications Review},
  5(1):3--55, 2001.

\bibitem{flajolet2006hidden}
Philippe Flajolet, Wojciech Szpankowski, and Brigitte Vall{\'e}e.
\newblock Hidden word statistics.
\newblock {\em Journal of the ACM (JACM)}, 53(1):147--183, 2006.

\bibitem{mitzenmacher2009survey}
Michael Mitzenmacher et~al.
\newblock A survey of results for deletion channels and related synchronization
  channels.
\newblock {\em Probability Surveys}, 6:1--33, 2009.

\bibitem{chase1976subsequence}
Phillip~J Chase.
\newblock Subsequence numbers and logarithmic concavity.
\newblock {\em Discrete Mathematics}, 16(2):123--140, 1976.

\bibitem{flaxman2004strings}
Abraham Flaxman, Aram~W Harrow, and Gregory~B Sorkin.
\newblock Strings with maximally many distinct subsequences and substrings.
\newblock {\em Electron. J. Combin}, 11(1):R8, 2004.

\bibitem{hirschberg1999bounds}
Daniel~S Hirschberg.
\newblock Bounds on the number of string subsequences.
\newblock In {\em Combinatorial Pattern Matching}, pages 115--122. Springer,
  1999.

\bibitem{hirschberg2000tight}
DANIEL~S Hirschberg and MIREILLE Regnier.
\newblock Tight bounds on the number of string subsequences.
\newblock {\em Journal of Discrete Algorithms}, 1(1):123--132, 2000.

\bibitem{jiang1995approximation}
Tao Jiang and Ming Li.
\newblock On the approximation of shortest common supersequences and longest
  common subsequences.
\newblock {\em SIAM Journal on Computing}, 24(5):1122--1139, 1995.

\bibitem{middendorf1995finding}
Martin Middendorf.
\newblock On finding minimal, maximal, and consistent sequences over a binary
  alphabet.
\newblock {\em Theoretical Computer Science}, 145(1):317--327, 1995.

\bibitem{middendorf2004combined}
Martin Middendorf and David~F Manlove.
\newblock Combined super-/substring and super-/subsequence problems.
\newblock {\em Theoretical computer science}, 320(2):247--267, 2004.

\bibitem{rahmann2006subsequence}
Sven Rahmann.
\newblock Subsequence combinatorics and applications to microarray production,
  dna sequencing and chaining algorithms.
\newblock In {\em Combinatorial Pattern Matching}, pages 153--164. Springer,
  2006.

\bibitem{elzinga2008algorithms}
Cees Elzinga, Sven Rahmann, and Hui Wang.
\newblock Algorithms for subsequence combinatorics.
\newblock {\em Theoretical Computer Science}, 409(3):394--404, 2008.

\bibitem{ullman1967capabilities}
Jeffrey~D Ullman.
\newblock On the capabilities of codes to correct synchronization errors.
\newblock {\em Information Theory, IEEE Transactions on}, 13(1):95--105, 1967.

\bibitem{swart2003note}
Theo~G Swart and Hendrik~C Ferreira.
\newblock A note on double insertion/deletion correcting codes.
\newblock {\em IEEE Transactions on Information Theory}, 49(1):269--273, 2003.

\bibitem{kanoria2013optimal}
Yashodhan Kanoria and Alessandro Montanari.
\newblock Optimal coding for the binary deletion channel with small deletion
  probability.
\newblock {\em Information Theory, IEEE Transactions on}, 59(10):6192--6219,
  2013.

\bibitem{graham2015binary}
Benjamin Graham.
\newblock A binary deletion channel with a fixed number of deletions.
\newblock {\em Combinatorics, Probability and Computing}, 24(03):486--489,
  2015.

\bibitem{sala2013counting}
Frederic Sala and Lara Dolecek.
\newblock Counting sequences obtained from the synchronization channel.
\newblock In {\em Information Theory Proceedings (ISIT), 2013 IEEE
  International Symposium on}, pages 2925--2929. IEEE, 2013.

\bibitem{sala2015three}
Frederic Sala, Ryan Gabrys, Clayton Schoeny, and Lara Dolecek.
\newblock Three novel combinatorial theorems for the insertion/deletion
  channel.
\newblock In {\em Information Theory (ISIT), 2015 IEEE International Symposium
  on}, pages 2702--2706. IEEE, 2015.

\bibitem{liron2015characterization}
Yuvalal Liron and Michael Langberg.
\newblock A characterization of the number of subsequences obtained via the
  deletion channel.
\newblock {\em Information Theory, IEEE Transactions on}, 61(5):2300--2312,
  2015.

\bibitem{cullina2012coloring}
Daniel Cullina, Ankur~A Kulkarni, and Negar Kiyavash.
\newblock A coloring approach to constructing deletion correcting codes from
  constant weight subgraphs.
\newblock In {\em Information Theory Proceedings (ISIT), 2012 IEEE
  International Symposium on}, pages 513--517. IEEE, 2012.

\bibitem{ordentlich2014bounding}
Or~Ordentlich and Ofer Shayevitz.
\newblock Bounding techniques for the intrinsic uncertainty of channels.
\newblock In {\em Information Theory (ISIT), 2014 IEEE International Symposium
  on}, pages 3082--3086. IEEE, 2014.

\bibitem{diggavi2007capacity}
Suhas Diggavi, Michael Mitzenmacher, and H~Pfister.
\newblock Capacity upper bounds for deletion channels.
\newblock In {\em Proceedings of the International Symposium on Information
  Theory}, pages 1716--1720, 2007.

\bibitem{kalai2010tight}
Adam Kalai, Michael Mitzenmacher, and Madhu Sudan.
\newblock Tight asymptotic bounds for the deletion channel with small deletion
  probabilities.
\newblock In {\em Information Theory Proceedings (ISIT), 2010 IEEE
  International Symposium on}, pages 997--1001. IEEE, 2010.

\bibitem{rahmati2013bounds}
Mehdi Rahmati and Tolga~M Duman.
\newblock Bounds on the capacity of random insertion and deletion-additive
  noise channels.
\newblock {\em Information Theory, IEEE Transactions on}, 59(9):5534--5546,
  2013.

\bibitem{cullina2014improvement}
Daniel Cullina and Negar Kiyavash.
\newblock An improvement to levenshtein's upper bound on the cardinality of
  deletion correcting codes.
\newblock {\em Information Theory, IEEE Transactions on}, 60(7):3862--3870,
  2014.

\bibitem{gentleman1989distribution}
Jane~F Gentleman and Ronald~C Mullin.
\newblock The distribution of the frequency of occurrence of nucleotide
  subsequences, based on their overlap capability.
\newblock {\em Biometrics}, pages 35--52, 1989.

\bibitem{wu2005distributions}
Chufang Wu.
\newblock The distributions of the frequency of occurrence of nucleotide
  subsequences.
\newblock {\em Methodology and Computing in Applied Probability},
  7(3):325--334, 2005.

\bibitem{drmota2012mutual}
Michael Drmota, Wojciech Szpankowski, and Krishnamurthy Viswanathan.
\newblock Mutual information for a deletion channel.
\newblock In {\em Information Theory Proceedings (ISIT), 2012 IEEE
  International Symposium on}, pages 2561--2565. IEEE, 2012.

\bibitem{flajolet2009analytic}
Philippe Flajolet and Robert Sedgewick.
\newblock {\em Analytic combinatorics}.
\newblock cambridge University press, 2009.

\bibitem{sedgewick2013introduction}
Robert Sedgewick and Philippe Flajolet.
\newblock {\em An introduction to the analysis of algorithms}.
\newblock Addison-Wesley, 2013.

\bibitem{guibas1981periods}
Leo~J Guibas and Andrew~M Odlyzko.
\newblock Periods in strings.
\newblock {\em Journal of Combinatorial Theory, Series A}, 30(1):19--42, 1981.

\bibitem{guibas1981string}
Leonidas~J Guibas and Andrew~M Odlyzko.
\newblock String overlaps, pattern matching, and nontransitive games.
\newblock {\em Journal of Combinatorial Theory, Series A}, 30(2):183--208,
  1981.

\end{thebibliography}

\end{document}